\documentclass{new_tlp}

\usepackage{amsmath}
\usepackage{multirow}

\newtheorem{theorem}{Theorem}
\newtheorem{definition}{Definition}

\title[Dynamic Consistency Checking]
	{Dynamic Consistency Checking in Goal-Directed Answer Set Programming}
\author[K. Marple and G. Gupta]
	{KYLE MARPLE and GOPAL GUPTA\\
	Department of Computer Science\\
	The University of Texas at Dallas}

\begin{document}
\maketitle

\begin{abstract}
In answer set programming, inconsistencies arise when the constraints placed on 
a program become unsatisfiable. In this paper, we introduce a technique for 
\textit{dynamic consistency checking} for our goal-directed method for computing
answer sets, under which only those constraints deemed relevant to the partial
answer set are tested, allowing inconsistent knowledgebases to be successfully
queried. However, the algorithm guarantees that, if a program has at least one
consistent answer set, any partial answer set returned will be a subset of some
consistent answer set. To appear in Theory and Practice of Logic Programming
(TPLP).
\end{abstract}


\section{Introduction}

Answer Set Programming (ASP) \cite{asp} has gained popularity as a way to 
develop non-monotonic reasoning applications.Three problems which prevent ASP 
from being adopted on a larger scale are (i) the need to compute a complete 
answer set regardless of the query, (ii) the ability of a minor inconsistency 
to render an entire knowledgebase useless, and (iii) the need to ground 
programs prior to execution. Our previous work with goal-directed ASP addresses
the first\cite{goalasp}, and we leave the third for future work. In this paper,
we address the second problem in context of goal-directed execution of answer
set programs. 

Currently, most popular ASP solvers rely on SAT solvers \cite{cmodels,clasp} 
which can't simply disregard inconsistencies that are unrelated to a query.
Because complete answer sets are computed, the underlying program must be
consistent. Thus much of the existing work in querying inconsistent 
knowledgebases has focused on repairing programs to restore consistency
\cite{arenas2003}. In contrast, our goal in this paper is to be able to work
with the consistent part of the knowledgebase, i.e., as long as a query does
not invoke clauses from the part of the knowledgebase that is inconsistent, we
should be able to execute it and produce an answer set, if one exists. Thus, we
do deviate from standard ASP semantics, as under ASP semantics, there are no
answer sets in the presence of inconsistencies in the knowledgebase.

In this paper, we introduce \textit{dynamic consistency checking} (DCC), a 
method for querying inconsistent databases that requires no modification of the
underlying programs or queries. Instead, DCC takes advantage of goal-directed
answer set programming to ignore inconsistencies that are unrelated to the
current query. Additionally, because DCC reduces the number of consistency
checks that a partial answer set must satisfy, it can significantly improve the
performance of goal-directed execution.

At the core of the problem is the issue of relevance. Because ASP and the
underlying stable model semantics lack a \textit{relevance property}, the
truth value of an atom can depend on other, totally unrelated rules and atoms
\cite{dix95}. Because such rules may not be encountered during normal top-down
execution, any goal-directed execution strategy for ASP must either alter the
semantics or employ some form of consistency checking to ensure correctness.
In designing our goal-directed method we chose the latter route, employing
consistency checks to ensure that constraints imposed by these rules are 
satisfied.

DCC employs splitting sets \cite{splitting} to reduce the number of 
consistency checks that must be satisfied while retaining strong guarantees 
regarding correctness. Execution using DCC employs a modified relevance 
criteria to determine which consistency checks are relevant to the current 
partial answer set, and only those checks are enforced.

DCC has been implemented as an extension of the \textit{Galliwasp} 
system \cite{galliwasp}, which makes use of our original goal-directed method.
As we will demonstrate, DCC has several advantages over other potential 
strategies based on ignoring unrelated inconsistencies. We will show that, if a
program has at least one consistent answer set, then a query will succeed using
DCC if and only if the partial answer set returned is a subset of some
consistent answer set. If no consistent answer set exists, then DCC can allow
partial answer sets to be found for a consistent subset of the program. We will
also demonstrate that DCC can improve the performance of goal-directed 
execution and that partial answer sets produced using DCC can provide more
targeted results than either full answer sets or partial answer sets with 
comprehensive consistency checking. 

The remainder of the paper is structured as follows.
In Section \ref{sec:challenges} we discuss issues that are potential 
impediments to widespread adoption of ASP. Next, in Section \ref{sec:goaldir},
we give an overview of goal-directed ASP, focusing on consistency checking. In
Section \ref{sec:dcc} we introduce our technique for dynamic consistency
checking using splitting sets and prove several interesting properties. In
Section \ref{sec:advantages} we examine advantages of DCC and compare the
results of \textit{Galliwasp} with and without dynamic consistency checking. 
Finally, in Section \ref{sec:conclusions} we discuss related and future work 
and draw conclusions.

\section{Answer Set Programming: Challenges} \label{sec:challenges}

While the Answer Set Programming paradigm has gained wide popularity among
researchers, there are still issues that stand in the way of its use by
ordinary users. The overarching goal of our research project is to eliminate
such issues. The major issues are briefly described next, though this paper
is mainly concerned with addressing only the last one.

The first problem relates to grounding an ASP program. Because existing 
systems for executing ASP programs rely on SAT solvers, ASP programs containing
predicates have to be grounded first. Even when restricted to finitely
groundable programs, the size of a grounded program can be exponentially large.
Thus, while writing ASP programs, one has to write code in a way that will
keep the size of the grounded program small. The grounding step can be avoided 
if goal-directed strategies, such as Galliwasp, are developed and used to
execute ASP programs. At present, even though the execution algorithm used by
Galliwasp is goal-directed, it assumes that the input program is grounded. Note
that work is in progress to extend Galliwasp so that predicate ASP programs
(including those containing functions) can be executed in a goal-directed
manner without being grounded \cite{elmer}. 

The second problem relates to computing an entire model of a program. Most 
current ASP execution methods compute the {\it entire} answer set, but in 
practice, we may only be interested in knowing if a specific piece of 
knowledge can be inferred. Consider the case of a large relational database 
coded in ASP. Without additional constraints, a complete answer set will 
contain all of the information in the database, not just the answer to a 
successful query. To work around this, constraints will need to be added to 
pare down the results, effectively requiring that a program be written where a 
single query might otherwise suffice. So if we rely on such solvers, then ASP
can be used for solving specific problems, but its use for building large
knowledge-based applications will pose challenges.

Finally, the third problem relates to being able to work with ASP programs 
which are inconsistent. As long as the answer being sought only 
depends on a consistent subset of the knowledgebase, one should be able to 
infer that knowledge. However, this is not the case with current ASP systems. 
The entire knowledgebase has to be consistent in order for them to produce a 
solution. To take a trivial example, consider a consistent ASP program 
$\kappa$ to which the clause {\tt p :- not p.} 
is added, where {\tt p} does not occur elsewhere in the program. The augmented 
program will have no answer sets. It is difficult for 
SAT solver-based approaches to identify subsets of the program that are 
consistent. A query-driven, goal-directed approach, in contrast, only `touches'
those parts of the program that are needed for establishing the query. All
constraints that involve any of the literals `touched' during the execution of 
the query, directly or indirectly, must also be enforced. However, constraints 
that do not involve such literals need not be executed, as they are independent
of the part of the program that was involved in answering the query. This is
precisely the idea behind our work on {\it dynamic consistency checking}
presented in this paper: only consistency checks that involve the portion of
the program that is `touched' by the query are executed. Thus, adding the rule
{\tt p :- not p.} to knowledgebase $\kappa$ above will not alter the execution
of the program unless a query contains \texttt{p}.

\section{Goal-Directed Answer Set Programming} \label{sec:goaldir}

Under our basic goal-directed method, a partial answer set is constructed by 
adding both positive and negative literals as they succeed during execution. 
When a query succeeds and all consistency checks have been satisfied, the set 
of positive literals in the partial answer set is guaranteed to be a subset of 
some consistent answer set of the program \cite{goalasp}. This method can be 
applied to arbitrary ASP programs, including those with rules that contain 
classical negation and disjunction. Such rules are simply converted to an 
equivalent set of normal rules.

Execution uses a modified form of co-SLD resolution (SLD resolution with 
coinduction) \cite{coinduction}. Under co-SLD resolution, each call is added to 
the \textit{coinductive hypothesis set} (CHS); a call can succeed 
{\it coinductively} if it unifies with an ancestor call in the CHS. In our 
goal-directed execution method, the CHS also serves as the candidate answer set.
However, some modifications are necessary to adapt co-SLD resolution to ASP:
\begin{itemize}
\item Negated calls are also allowed to succeed coinductively, i.e., negated
calls (e.g., not p) are added to the CHS. A negated call can
succeed coinductively if it unifies with an ancestor negated call in the CHS.
\item A literal and its negation cannot be in the CHS at the same time. If 
adding a literal to the CHS leads to such a situation, the computation fails 
and backtracking ensues.
\item Coinductive success is allowed only if an even, non-zero number of 
 negations occur between the recursive call and its ancestor call.
\end{itemize}

While the above description covers the basic execution of our algorithm, it 
omits perhaps the most important part, consistency checking. To understand the 
role of consistency checking, we must first examine the issue of relevance in 
more detail.

\subsection{Relevance} \label{sec:relevance}

The issue of relevance is central to goal-directed ASP. In defining relevance,
\cite{dix95} uses the dependency graph of a program $P$ and the following
notions:
\begin{itemize}
\item ``$dependencies\_of(X) := \{ A : X$ depends on $A \}$'', i.e. $X$ calls 
$A$ directly or indirectly, and
\item ``$rel\_rule(P,X)$ is the set of \textit{relevant rules} of $P$ with 
respect to $X$, i.e. the set of rules that contain an
$A \in dependencies\_of(X)$ in their heads.''
\end{itemize}
Then, ``given any semantics SEM and a program $P$, it is perfectly reasonable 
that the truth-value of a literal $L$, with respect to SEM($P$), only depends 
on the subprogram formed from the relevant rules of $P$ with respect to $L$'', 
formalized as:
\begin{definition} \label{def:relevance}
``Relevance states that for all literals $L$:
$SEM(P)(L) = SEM(rel\_rule(P,L))(L )$.'' \cite{dix95}
\end{definition}

Despite being ``perfectly reasonable'', the above definition of relevance does
not hold for ASP. This is due to the presence of rules which contain an odd
loop over negation (OLON). OLONs occur implicitly in rules with an empty head,
but also occur in rules with non-empty heads, whenever a rule can be called
recursively with an odd number of negations between the original and recursive
calls. These ``OLON rules'' place constraints on a program that must be
satisfied by any consistent answer set. For example, given an OLON rule of the
form:
\begin{verbatim}
p :- B, not p.
\end{verbatim}
where \texttt{B} is a conjunction of literals, one of the following 
must be satisfied:
\begin{enumerate}
\item \texttt{p} must succeed through other means, or
\item at least one literal in \texttt{B} must fail.
\end{enumerate}
That is, the rule imposes the constraint $p \vee not\;B$ on the program. Such a
rule can thus alter the truth-value of a literal in \texttt{B} despite not
being relevant to the literal under Definition \ref{def:relevance}.

\subsection{Consistency Checking} \label{sec:consistency}

Because ASP lacks a relevance property of its own, our algorithm uses
consistency checks to enforce a modified relevance property, where for a
program $P$ and literal $L$, the set of rules in $P$ relevant to $L$ is
expanded to include every OLON rule in the program \cite{goalasp}. That is,
\begin{equation} \label{eq:nmrrelrule}
\begin{aligned}
nmr\_rel\_rul(P,L) = rel\_rul(P,L) \cup OLON(P)
\end{aligned}
\end{equation}

\noindent where rel\_rul($P$,$L$) is the set of relevant rules defined in 
Section \ref{sec:relevance} and OLON(P) is the set of OLON rules in $P$. The 
semantics of $P$ with respect to $L$ can now be defined in terms of the
subprogram formed by the expanded set of relevant rules:
\begin{equation}
\label{eq:nmrrel}
SEM(P)(L) = SEM(nmr\_rel\_rule(P,L))(L)
\end{equation}

\noindent This property ensures that an answer set of the subprogram, if one
exists, will be a subset of some consistent answer set of $P$ \cite{goalasp}.

To enforce our modified relevance property, our method uses a special rule, the
\textit{non-monotonic reasoning check} (NMR check), which calls a sub-check for
each OLON rule in a program. Each sub-check ensures that the associated OLON
rule is satisfied. The NMR check is then automatically appended to each query,
ensuring that the property will hold for any query which succeeds
\cite{goalasp}.

\begin{figure}
\figrule
\begin{verbatim}
p :- q.               % Rule 1: OLON
q :- not r, not p.    % Rule 2: OLON
r :- not p.           % Rule 3: Ordinary
:- q, r.              % Rule 4: OLON

chk_1 :- p.
chk_1 :- not q.
chk_2 :- r.
chk_2 :- p.
chk_2 :- q.
chk_4 :- not q.
chk_4 :- not r.
nmr_check :- chk_1, chk_2, chk_4.
\end{verbatim}
\caption{A simple program with consistency checks added.} \label{fig:nmrexample}
\figrule
\end{figure}

The construction of the sub-checks involves creating rules for the dual of each 
OLON rule in the program. Duals explicitly encode the negation of a literal.
For example, given:
\begin{verbatim}
p :- q, not r.
\end{verbatim}
\noindent the dual rules for \texttt{p} are:
\begin{verbatim}
not p :- not q.
not p :- r.
\end{verbatim}
\noindent In the case of sub-checks, the negation of the head is first appended 
to encapsulate the success of a literal through other means. The entire process 
is as follows:
\begin{enumerate}
\item For rules with non-empty heads, the negation of the head is appended to 
	the body of the rule, if not already present.
\item The dual of the rule is computed.
\item The dual is given a unique head, which is also added to the body of the 
	NMR sub-check.
\end{enumerate}
The example in Figure \ref{fig:nmrexample} shows a simple ASP program with the 
NMR check and sub-checks added.

While this ensures that our method adheres to the semantics of ASP, the
execution of the NMR check can adversely impact performance. ASP programs
routinely make heavy use of headless rules to enforce constraints, which can
result in an NMR check which contains thousands of goals. For example, an
instance of the 20-Queens problem can produce an NMR check containing 25,100
goals, each representing a sub-check that must be executed alongside any query.
A means of reducing the performance impact of these checks is thus extremely
desirable.

\section{Dynamic Consistency Checking} \label{sec:dcc}

Dynamic Consistency Checking (DCC) began as an attempt to improve the 
performance of goal-directed execution. While we have developed various other 
techniques to reduce the performance impact of consistency checking, none of
them reduce the actual number of checks that must be satisfied, as this is
impossible to do while guaranteeing full compliance with the ASP semantics. DCC
was our attempt to reduce the number of checks performed while staying as close
to the original ASP semantics as possible.

As any reduction in the number of consistency checks will result in 
non-compliance with the ASP semantics, selecting which checks to enforce 
depends on the properties desired from the modified semantics. In the case of 
DCC, these properties also make the technique useful for querying inconsistent
knowledgebases.
\begin{definition} \label{def:properties}
For a program P, the \textit{desired properties} of DCC are:
\begin{enumerate}
\item Execution shall always be consistent with the ASP semantics of the
	sub-program of P (further defined in Section \ref{sec:dccrelevance}).
\item If P has at least one consistent answer set, execution shall be 
	consistent with the ASP semantics of P.
\end{enumerate}
\end{definition}

In this section, we discuss the relevance property employed by DCC before 
moving on to the algorithm itself. Finally, we provide proofs that DCC 
satisfies the above properties.

\subsection{Relevance Under DCC} \label{sec:dccrelevance}

While our original relevance property, given in Formula \ref{eq:nmrrel}, makes
every consistency check relevant to every literal, DCC selects only those
checks necessary to enforce our desired properties from Definition 
\ref{def:properties}. Relevant checks are dynamically selected based on the 
literals in the partial answer set.

\begin{figure}
\figrule
\begin{verbatim}
a :- b.
b :- not c.
c :- not b.

p :- a.
q :- b.
:- p, q.

chk_1 :- not p.
chk_1 :- not q.
nmr_check :- chk_1.
\end{verbatim}
\caption{Example program (consistency checks added).} \label{fig:example1}
\figrule
\end{figure}

At first glance, it might seem sufficient to select only those checks which 
directly call literals in the partial answer set (or their negations). However, 
this can lead to incorrect results. Consider the program in Figure 
\ref{fig:example1}. One consistent answer set exists:
$\{ c, not\,a, not\,b, not\,p, not\,q \}$.
However, given a query \texttt{?- a.}, selecting only those checks which
directly call some literal in the partial answer set will yield
$\{ a, b, not\,c \}$, thus violating our desired properties.

\begin{figure}
\figrule
\begin{verbatim}
:- p, q.
q :- not r, not q.

chk_1 :- not p.
chk_1 :- not q.
chk_2 :- r.
chk_2 :- q.
nmr_check :- chk_1, chk_2.
\end{verbatim}
\caption{Example program (consistency checks added).} \label{fig:example2}
\figrule
\end{figure}

Clearly, our properties require that we select at least those checks which can 
potentially reach a literal in the partial answer set. However, this can lead 
to behavior that is difficult to predict. Consider the program in Figure 
\ref{fig:example2} with the query \texttt{?- not p.} The presence of either
\texttt{q} or \texttt{not q} in each OLON rule might seem to indicate that both
consistency checks will be activated and cause the query to fail. However, only
\texttt{chk\_1} will be activated. Because the first clause will succeed,
neither \texttt{q} nor its negation will be added to the partial answer set,
and the query will succeed.

To achieve more predictable behavior, DCC selects relevant checks using 
specially constructed \textit{splitting sets}. A splitting set for a program is 
any set of literals such that if the head of a rule is in the set, then every 
literal in the body of the rule must also be in the set \cite{splitting}. The 
rules in a program $ P $ can then be divided relative to a splitting set 
\textit{U} into the bottom, $b_U(P)$, containing those rules whose head is in
\textit{U}, and the top, $P \setminus b_{U}(P)$.

The splitting sets used to determine relevant NMR sub-checks are created by 
constructing splitting sets for each NMR sub-check and merging sets whose 
intersection is non-empty. The result is a set of disjoint splitting sets 
$U_{i}$ such that for an NMR sub-check \textit{C}, if $C \in U_{i}$, then for 
every literal \textit{L} reachable by \textit{C}, $L \in U_{i}$. This allows us 
to define the sub-checks relevant to a literal as those whose heads are in the 
same splitting set:
\begin{equation}
\label{eq:dccrelrule}
\begin{aligned}
dcc\_rel\_rul(P,L) = rel\_rul(P,L) \cup OLON(P,L),\\
OLON(P,L) =
	\lbrace \mbox{$R: R \in OLON(P) \cap b_{U_i}(P) \wedge L \in U_i$} \rbrace
\end{aligned}
\end{equation}

\noindent where OLON(P,L) is the set of OLON rules relevant to $L$. This leads
us to DCC's relevance property, which defines the semantics of $P$ with respect
to $L$ in terms of the new set of relevant rules:
\begin{equation}
\label{eq:dccrel}
SEM(P)(L) = SEM(dcc\_rel\_rule(P,L))(L)
\end{equation}

This definition allows for more predictable behavior than simply selecting the 
checks reachable by a given literal. In the case of Programs \ref{fig:example1} 
and \ref{fig:example2}, only one splitting set will be created, resulting in 
behavior that is identical to normal goal-directed ASP. Indeed, as we will 
prove in Section \ref{sec:proof}, execution will be consistent with ASP 
whenever a program has at least one answer set.

\subsection{Execution with DCC}

Given DCC's relevance property in Formula \ref{eq:dccrel}, our goal-directed 
execution strategy must be modified to enforce it. A query should succeed if and
only if every OLON rule relevant to a literal in the partial answer set is
satisfied. In addition to creating the associated splitting sets, the
application of the relevant NMR sub-checks also becomes more complex.

The creation of the necessary splitting sets can be accomplished by examining a 
program's call graph after the NMR sub-checks have been added. A simple 
depth-first search is sufficient to construct the splitting set for an
individual sub-check, after which overlapping sets can be merged. For added
efficiency, constructing and merging the sets can be performed simultaneously:
whenever a literal is encountered that has already been added to another set,
that set is merged with the current one. This eliminates the need to traverse
any branch in the call graph more than once. The overhead of searching the sets
themselves can be minimized with proper indexing.

To apply the NMR check when executing a query with DCC, it must also be 
dynamically constructed. The NMR check should consist of those sub-checks which 
are relevant to a literal in the partial answer set. However, because the 
sub-checks themselves may add literals to the partial answer set, simply 
executing the query and then selecting the relevant checks once is 
insufficient. Instead, each time a literal succeeds, the relevant sub-checks 
are added to the NMR check. Similarly, the state of the NMR check is restored 
when backtracking occurs. In this manner, the NMR check will always remain 
consistent with the current partial answer set.

\subsection{Correctness of DCC} \label{sec:proof}

Now that we have established DCC's algorithm, we can prove that it satisfies 
the property it was designed to enforce. That is:

\begin{theorem}
If a program $P$ has at least one consistent answer set, then a query will
succeed under DCC if and only if the partial answer set is a subset of some
consistent answer set of $P$.
\end{theorem}

\begin{proof}
Observe that, if a DCC query succeeds, the partial answer set will be 
$X = A \cup B$ where
\begin{itemize}
\item $A$ is a partial answer set of the splitting set $U$ formed by the union
	of the splitting sets containing relevant NMR sub-checks
\item $B$ is the set of succeeding literals which are not reachable by any NMR
	sub-check
\end{itemize}

\noindent Per the Splitting Set Theorem \cite{splitting}, a set $X'$ is an 
answer set of $P$ if and only if $X' = A' \cup B'$ where $A'$ is an answer set
of $b_{U}(P)$, $B'$ is an answer set of $e_U(P \setminus b_{U}(P), A')$, and
$A' \cup B'$ is consistent.\footnote{For a set $X$ of positive literals in $U$,
$e_U(P \setminus b_{U}(P), X)$ is a partial evaluation of the top of $P$ with
respect to $X$. The partial evaluation is constructed by first dropping rules
whose bodies contain the negation of a literal in $X$ and them removing calls
to literals in $X$ from the bodies of the remaining rules.} Thus our theory
will hold if $A \subseteq A'$, $B \subseteq B'$ and $A' \cup B'$ is consistent.

Because every NMR sub-check relevant to some literal in $A$ will be activated 
and must succeed for the DCC query to succeed, $A$ will always be a subset of 
some consistent answer set of $b_{U}(P)$. Furthermore, such an answer set must
exist for the DCC query to succeed. Thus, for any succeeding DCC query, there
exists an answer set $A'$ of $b_{U}(P)$ such that $A \subseteq A'$.\footnote{If
no literals in the query are reachable by any NMR sub-checks, $U$ will be empty
and both $A'$ and $A$ will be the empty set.}

Because only OLON rules can lead to inconsistency in an ASP
program\footnote{While rules involving classical negation and disjunction can
lead to inconsistency, Galliwasp handles these by converting them to a set of
equivalent normal rules, including OLON rules.}, the set $B$ will always be a
subset of some consistent answer set of $e_U(P \setminus b_{U}(P), A')$, if one
exists. Therefore, if at least one consistent answer set exists for $P$, we can
select $B'$ such that $B'$ is an answer set of $e_U(P \setminus b_{U}(P), A)$ 
such that $B \subseteq B'$.

Finally, because $A'$ contains every NMR sub-check relevant to any literal in
$A$, $A'$ will always be consistent with $B'$. Thus, if $P$ has at least one
answer set, a query will succeed under DCC if and only the partial answer set
is a subset of some consistent answer set of $P$.
\end{proof}

\section{Advantages of DCC} \label{sec:advantages}

Execution with DCC offers several advantages over normal goal-directed ASP. The 
three primary advantages are partial answer sets of inconsistent programs, 
output that is relevant to the query, and improved performance.

\subsection{Answer Sets of Inconsistent Programs}

One disadvantage of ASP is the way in which it handles inconsistency in a 
knowledgebase. Any inconsistency, no matter how small, renders the entire 
program inconsistent, and thus no answer set will exist. This behavior can be 
particularly inconvenient in large knowledgebases where an inconsistency may be 
completely unrelated to a particular query. Given a large, perfectly consistent 
database implemented in ASP, adding the rule \verb|:- not c.| where \texttt{c} 
is a unique literal, will cause any query to the database to fail.

With DCC, if a query succeeds prior to adding the rule above, then it will 
continue to succeed even after the rule is added.

\subsection{Query-relevant Output} \label{sec:outputrelevance}

One advantage of goal-directed ASP is the ability to compute partial answer 
sets using a query.
Ideally, partial answer sets will contain only literals which are related to 
the query. However, the execution of the NMR check can force the addition of
literals which are unrelated to the current query. By omitting unnecessary NMR
checks, DCC can limit this irrelevant output.

Consider the case where two consistent ASP programs, A and B, are concatenated 
to form a new program C. Assume that A and B have no literals in common and 
that each contains one or more OLON rules. A full answer set of C will 
obviously contain literals from both of the sub-programs. As a result of the
OLON rules, any partial answer set obtained using goal-directed ASP will also
contain literals from both sub-programs. However, using DCC, a succeeding query
which targets only one sub-program will only contain literals from that
sub-program.

Exploiting this behavior does require care on the part of the programmer. For 
example, many ASP programs use OLON rules in place of queries. However, such 
rules will often force all or most of a program's literals into a single 
splitting set. As a result, every OLON rule will always be deemed relevant, and 
DCC will function no differently than normal goal-directed ASP. We will see
this behavior in some of the sub-programs examined in the next section.

\subsection{Performance Compared to Normal Consistency Checking} 
\label{sec:performance}

\begin{table}
	\caption{Comparative Performance Results}
	\label{tab:perf}
	\begin{minipage}{\textwidth}
	\begin{tabular}{lclrr}
		\hline\hline
		\multicolumn{1}{l}{\multirow{2}{*}{Problem}} &
		\multicolumn{1}{c}{\multirow{2}{*}{Splitting Sets}} &
		\multicolumn{1}{l}{\multirow{2}{*}{Query}} &
		\multicolumn{2}{c}{\multirow{1}{*}{Execution Times\footnote{CPU time in seconds.}}} \\
		& & & Original & w/ DCC \\
		\noalign{\smallskip}
		\hline
		\noalign{\smallskip}
		hanoi-5x15    & 0 & solveh & 0.276 & 0.274\\
		pigeons-30x30 & 1 & solvep & 0.065 & 0.065\\
		schur-3x13    & 1 & solves & 0.105 & 0.105\\
		hanoi-schur   & 1 & solveh & 0.134 & 0.028\\
		hanoi-schur   & 1 & solves & 0.134 & 0.134\\
		hanoi-pigeons & 1 & solveh & 0.346 & 0.341\\
		hanoi-pigeons & 1 & solvep & 0.343 & 0.342\\
		pigeons-schur & 2 & solvep & 9.958 & 0.672\\
		pigeons-schur & 2 & solves & 9.745 & 0.172\\
		han-sch-pigs  & 2 & solveh & 9.817 & 0.093\\
		han-sch-pigs  & 2 & solvep & 9.780 & 0.094\\
		han-sch-pigs  & 2 & solves & 9.942 & 0.201\\
		\hline\hline
	\end{tabular}
	\vspace{-0.5\baselineskip}
	\end{minipage}
\end{table}

In this section we compare \textit{Galliwasp's} performance on several programs,
with and without DCC. As the results in Table \ref{tab:perf} demonstrate, 
programs that take advantage of DCC can see a massive improvement in 
performance. Additionally, even when a program does not take advantage of DCC, 
the overhead remains minimal.

To simulate programs which take advantage of DCC, the following three programs 
were concatenated together in various combinations:

\begin{itemize}
\item \texttt{hanoi-5x15} is a 5 ring, 15 move instance of the Towers of Hanoi.
	The query \texttt{?- solveh.} will return a partial answer set containing
	the solution.
\item \texttt{pigeons-30x30} is an instance of the MxN-Pigeons problem. The
	query \texttt{?- solvep.} will find a complete answer set.
\item \texttt{schur-3x13} is a 3 partition, 13 number instance of the Schur
	Numbers problem. The query \texttt{?- solves.} finds a complete answer set.
\end{itemize}

\noindent Each of the three base programs, and thus each combination, has at 
least one consistent answer set. The Towers of Hanoi instance contains no OLON 
rules, and consequently no splitting sets. The other two programs contain OLON 
rules that force the computation of a complete answer set, and thus have one 
splitting set each. As a result, a DCC query containing only \texttt{solveh} 
will not activate any NMR sub-checks, while queries containing \texttt{solvep} 
or \texttt{solves} will activate every NMR sub-check for their respective 
problems. Thus DCC execution of \texttt{solveh} will not access any splitting 
sets, while \texttt{solvep} and \texttt{solves} will access one set each.

In general, the fewer splitting sets accessed by a DCC query relative to the 
total, the better it will perform compared to a non-DCC query. This is
exemplified by the cases with two splitting sets in Table \ref{tab:perf}. In
the programs tested, each splitting set represents a large number of OLON
rules. As the non-DCC results indicate, the negative impact of increasing the
number of OLON rules can be immense. DCC is able to avoid this by satisfying
only those rules relevant to the current query.

\section{Related and Future Work} \label{sec:related}

DCC is an extension of goal-directed ASP \cite{goalasp} and has been implemented
using the \textit{Galliwasp} system \cite{galliwasp}. The technique relies 
heavily on the properties of splitting sets, and the Splitting Set Theorem in 
particular \cite{splitting}.

Numerous other methods for querying inconsistent databases have been developed. 
The problem of Consistent Query Answering is defined in terms of minimal 
database repairs in \cite{arenas1999}, which develops a technique based on 
query modification that is built upon in several subsequent works 
\cite{celle2000,arenas2003}. 
However, these techniques require that database inconsistencies be identified 
and accounted for. Because DCC relies on a goal-directed technique for 
computing answer sets, our method allows inconsistent information to simply be 
ignored unless it directly relates to the current query.

Plans for future work focus on modifying the technique to work with ungrounded 
ASP programs. Detecting OLONs and constructing the associated splitting sets 
prior to grounding has the potential to both reduce the overhead and allow the 
use of DCC with a wider range of solvers. Of particular interest is integration 
with a datalog ASP system currently under development \cite{elmer}.

\section{Conclusions} \label{sec:conclusions}

In this paper we have introduced Dynamic Consistency Checking (DCC), a 
technique for querying inconsistent ASP programs using a goal-directed 
execution method. We have discussed the relevant aspects of goal-directed ASP, 
presented the relevance criteria which DCC enforces and proven that DCC is 
consistent with the ASP semantics for programs which have at least one 
consistent answer set. Additionally, we have examined the advantages of DCC 
with respect to querying inconsistent databases, achieving more useful output 
from queries, and improving the performance of the \textit{Galliwasp} system. 
As our results demonstrate, DCC can be efficiently implemented and programs 
which take advantage of it can achieve significant benefits. Future work will 
focus on allowing DCC to operate on ungrounded ASP programs and adapting the 
technique into additional ASP solvers.

\bibliographystyle{acmtrans}
\bibliography{dynamic_consistency_checking}

\end{document}